\documentclass[preprint,12pt]{elsarticle}

\pdfoutput=1
\usepackage{graphicx,epsfig}
\usepackage{array}
\usepackage{amsmath}
\usepackage{amsfonts}
\usepackage{amssymb}
\usepackage{float}
\biboptions{sort&compress}

\newtheorem{theorem}{Theorem}[section]
\newtheorem{lemma}[theorem]{Lemma}
\newtheorem{definition}[theorem]{Definition}

\newtheorem{remark}[theorem]{Remark}

\newtheorem{assumption}[theorem]{Assumption}
\newproof{proof}{Proof}
\newcommand{\sr}{\stackrel}

\newcommand{\rar}{\rightarrow}

\newcommand{\tri}{\sr{\triangle}{=}}

\newcommand{\be}{\begin{equation}}
\newcommand{\ee}{\end{equation}}
\newcommand{\bea}{\begin{eqnarray}}
\newcommand{\eea}{\end{eqnarray}}
\newcommand{\bes}{\begin{eqnarray*}}
\newcommand{\ees}{\end{eqnarray*}}
\newcommand{\bp}{\begin{problem}}
\newcommand{\ep}{\end{problem}}

\newcommand{\noi}{\noindent}

\newcommand{\bc}{\begin{center}}
\newcommand{\ec}{\end{center}}

\journal{Automatica}

\begin{document}

\begin{frontmatter}


 \title{Nonanticipative Rate Distortion Function and Filtering Theory:~A Weak Convergence Approach\tnoteref{label1}}
\tnotetext[label1]{This work was financially supported by a medium size University of Cyprus grant entitled “DIMITRIS” and by European Community's Seventh Framework Programme (FP7/2007-2013) under grant agreement no. INFSO-ICT-223844. Part of this work was presented in 12$^{th}$ Biannual European Control Conference (ECC '13) \cite{charalambous-stavrou2013bb}.}
\author[rvt]{P.~A.~Stavrou}
\ead{stavrou.fotios@ucy.ac.cy}
\ead[url]{http://www.photiosfstavrou.com}
\author[rvt]{C.~D.~Charalambous\corref{cor1}}
\ead{chadcha@ucy.ac.cy}
\ead[url]{http://www.eng.ucy.ac.cy/chadcha/}
\cortext[cor1]{Corresponding author at: Electrical and Computer Engineering Department, University of Cyprus, Nicosia, Cyprus}
\address[rvt]{The authors are with the Department of Electrical and Computer Engineering (ECE), University of Cyprus, Nicosia, Cyprus}

%

\begin{abstract}
In this paper the relation between nonanticipative rate distortion function (RDF) and Bayesian filtering theory is further investigated on general Polish spaces. The relation is established via an optimization on the space of conditional distributions of the so-called directed information subject to fidelity constraints. Existence of the optimal reproduction distribution of the nonanticipative RDF is shown using the topology of weak convergence of probability measures. Subsequently, we use the solution of the nonanticipative RDF to present the realization of a multidimensional partially observable source over a scalar Gaussian channel. We show that linear encoders are optimal, establishing joint source-channel coding in real-time. 
\end{abstract}


\end{frontmatter}


\section{Introduction}\label{introduction}
\par  In the past, rate distortion (or distortion rate) functions and filtering theory have evolved independently. Specifically, classical rate distortion function (RDF) addresses the problem of reproduction of a process subject to a fidelity criterion without much emphasis on the realization of the reproduction conditional distribution via causal\footnote{The terms causal and nonanticipative are used interchangeably with the same meaning for conditional distributions.} operations. On the other hand, filtering theory is developed by imposing real-time realizability on estimators with respect to measurement data. Specifically, least-squares filtering theory deals with the characterization of the conditional distribution of the unobserved process given the measurement data, via a stochastic differential equation which causally depends on the observation data \cite{elliott-aggoun-moore1995}.\\
  Although, both reliable communication and filtering (state estimation for control) are concerned with the reproduction of processes, the main underlying assumptions characterizing them are different.\\ 
\noi Historically, the work of R. Bucy  \cite{bucy} appears to be the first to consider the direct relation between distortion rate function and filtering, by carrying out the computation of a realizable distortion rate function with square criteria for two samples of the Ornstein-Uhlenbeck process. The work of A. K. Gorbunov and M. S. Pinsker \cite{gorbunov91} on $\epsilon$-entropy defined via a causal constraint on the reproduction distribution of the RDF, although not directly related to the realizability question pursued by  Bucy, computes the nonanticipative RDF for stationary Gaussian processes via power spectral densities. Recently, the authors in \cite{charalambous-stavrou-ahmed2013} investigated relations between filtering theory and RDF defined via mutual information using the topology of weak$^*$ convergence on appropriate defined spaces. The derivations of the results in \cite{charalambous-stavrou-ahmed2013} require elaborate arguments.\\
\noi The objective of this paper is to further investigate the connection between nonanticipative rate distortion theory and filtering theory for general distortion functions and random processes on abstract Polish spaces using the topology of weak convergence. Moreover, instead of mutual information we invoke directed information with an inherent causality, which defines the reproduction conditional distribution. Further, the connection is established  via optimization of directed information \cite{massey90} over the space of conditional distributions which satisfy an average distortion constraint. In comparison to \cite{charalambous-stavrou-ahmed2013}, we impose natural technical assumptions to obtain analogous results under the topology of weak convergence of probability measures, by using Prohorov's theorem without introducing new spaces as done in \cite{charalambous-stavrou-ahmed2013}. We also present a new example to illustrate the realization of the filter via nonanticipative RDF. Specifically, we consider a multidimensional partially observable source, we compute the nonanticipative RDF, and we show how to realize it over a scalar additive Gaussian noise channel showing that linear encoder strategies are optimal. This example is new and it is considered as an open problem in information theory \cite{derpich-ostergaard2012}.\\
The main results discussed in this paper are the following.
\begin{description}
\item[(1)] Existence of optimal reproduction distribution minimizing directed information using the topology of weak convergence of probability measures on Polish spaces;
\item[(2)] example of a multidimensional source which is realized over a scalar additive Gaussian noise channel, for which the filter is obtained.
\end{description}
\noi This work is motivated by recent applications of sensor networks in which estimators are desired to have a specific accuracy, when processing information \cite{,gupta-dana-hespanha-murray-hassibi2009,yuksel2011}, and control over limited rate communication channel applications \cite{tatikonda-mitter2004,freudenberg-middleton2008,yuksel-meyn2013}. It is important to note that over the years several papers have appeared in the literature utilizing information theoretic measures for estimator and control applications \cite{feng-loparo-fang1997,guo-yin-wang-chai2009}.\\
First, we give a brief high level discussion on the relation between nonanticipative RDF and filtering theory, and discuss their connection.\\
Consider a discrete-time process $X^n\tri\{X_0,X_1,\ldots,X_n\}\in{\cal X}_{0,n} \tri \times_{i=0}^n{\cal X}_i$, and its reproduction $Y^n\tri\{Y_0,Y_1,\ldots,Y_n\}\in{\cal Y}_{0,n} \tri \times_{i=0}^n{\cal Y}_i$ where ${\cal X}_i$ and ${\cal Y}_i$ are Polish spaces.
\\

\noi{\it Bayesian Estimation Theory.} In classical filtering (see Fig.~\ref{filtering}), one is given a mathematical model that generates the process $X^n$, $\{P_{X_i|X^{i-1}}(dx_i|x^{i-1}):i=0,1,\ldots,n\}$, often induced via discrete-time recursive dynamics, a mathematical model that generates observed data obtained from sensors, say $Z^n$, $\{P_{Z_i|Z^{i-1},X^i}$ $(dz_i|z^{i-1},x^i):i=0,1,\ldots,n\}$, and the objective is to compute causal estimates of some function of the process $X^n$ based on the observed data $Z^n$.
The classical Kalman Filter is a well-known example, where the estimate $\widehat{X}_i =\mathbb{E}[X_i | Z^{i-1}],~i=0,1,\ldots,n$, is the conditional mean which minimizes the average least-squares estimation error. 	
\begin{figure}[ht]
\centering
\includegraphics[scale=0.75]{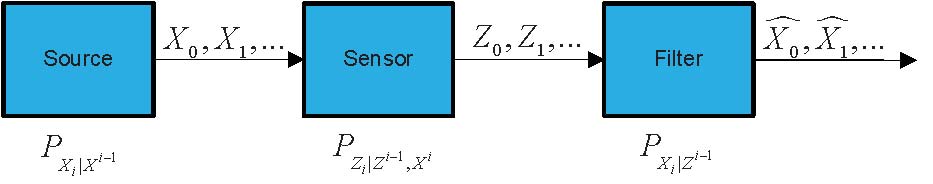}
\caption{Block Diagram of the Filtering Problem}
\label{filtering}
\end{figure}

\noi{\it Nonanticipative Rate Distortion Theory and Estimation.} In nonanticipative rate distortion theory one is given a process $X^n$, which induces a distribution $\{P_{X_i|X^{i-1}}(dx_i|x^{i-1}):~i=0,1,\ldots,n\}$, and the objective is to determine the causal reproduction conditional distribution $\{P_{Y_i|Y^{i-1},X^i}(dy_i|y^{i-1},x^i):~i=0,1,\ldots,n\}$ which minimizes the directed information from $X^n$ to $Y^n$ subject to distortion or fidelity constraint. The filter $\{Y_i:~i=0,1,\ldots,n\}$ of $\{X_i:~i=0,1,\ldots,n\}$ is found by realizing the optimal reproduction distribution $\{P_{Y_i|X^{i-1},X^i}(dy_i|y^{i-1},x^i):~i=0,1,\ldots,n\}$ via a cascade of sub-systems as shown in Fig. 2. Thus, in nonanticipative rate distortion theory the observation or mapping from $\{X_i:~i=0,1,\ldots,n\}$ to $\{Z_i:~i=0,1,\ldots,n\}$ is part of the realization procedure, while in filtering theory, this mapping is given \'a priori. Indeed, this is the main difference between Bayesian estimation theory and nonanticipative RDF for the purpose of estimation.
\begin{figure}[ht]
\vspace*{0.1cm}
\centering
\includegraphics[scale=0.75]{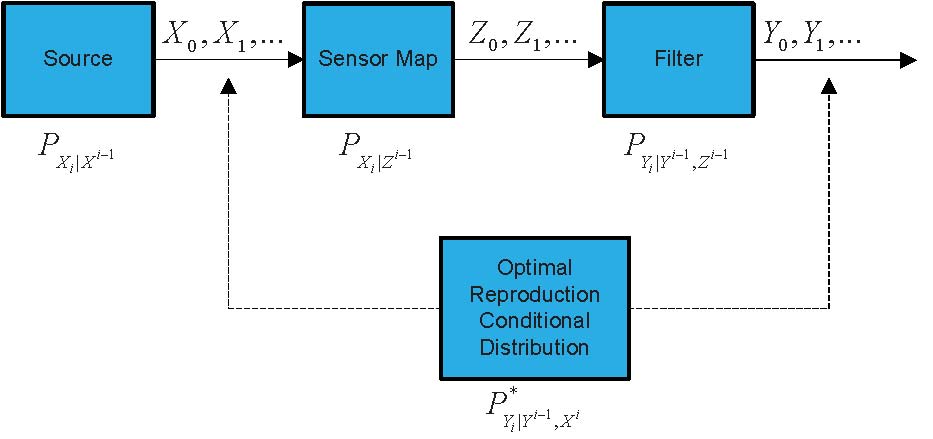}
\caption{Block Diagram of Filtering via Nonanticipative Rate Distortion Function}
\label{filtering_and_causal}
\end{figure}

\noi The precise problem formulation necessitates  the definitions of  distortion function or fidelity, and directed information.\\
The distortion function or fidelity constraint \cite{berger} between $x^n$ and its reproduction $y^n$, is a measurable function defined by
\begin{align*}
d_{0,n} : {\cal X}_{0,n} \times {\cal Y}_{0,n} \rar [0, \infty], \: \: d_{0,n}(x^n,y^n)\tri\frac{1}{n+1}\sum^n_{i=0}\rho_{0,i}(x^i,y^i).
\end{align*}
\noi Directed information  from a sequence of Random Variables (RV's) $X^n\tri\{X_0,X_1,\ldots,X_n\}\in{\cal X}_{0,n}\tri\times_{i=0}^n{\cal X}_i$, to another sequence $Y^n\tri\{Y_0,Y_1,\ldots,Y_n\}\in{\cal Y}_{0,n}\tri\times_{i=0}^n{\cal Y}_i$ is often defined via \cite{massey90,charalambous-stavrou2012}\footnote[4]{Unless otherwise, integrals with respect to probability distributions are over the spaces on which these are defined.}  
\begin{align}
I(X^n\rightarrow{Y}^n)&\tri\sum_{i=0}^n{I}(X^i;Y_i|Y^{i-1})\nonumber\\
&=\sum_{i=0}^n\int\log\Big(\frac{P_{Y_i|Y^{i-1},X^i}(dy_i|y^{i-1},x^i)}{{P}_{Y_i|Y^{i-1}}(dy_i|y^{i-1})}\Big)P_{X^i,Y^i}(dx^i,dy^i)\label{1a}\\
&\equiv\mathbb{I}_{X^n\rightarrow{Y^n}}(P_{X_i|X^{i-1},Y^{i-1}},P_{Y_i|Y^{i-1},X^i}:~i=0,1,\ldots,n).\label{1b}
\end{align}
The notation $\mathbb{I}_{X^n\rightarrow{Y^n}}(\cdot,\cdot)$ illustrates the dependence of directed information $I(X^n\rightarrow{Y^n})$ on the two sequences of nonanticipative or causal conditional distributions $\{P_{X_i|X^{i-1},Y^{i-1}}(\cdot|\cdot,\cdot),~P_{Y_i|Y^{i-1},X^i}(\cdot|\cdot,\cdot)~:~i=0,1,\ldots,n\}$. In information theory, directed information $\mathbb{I}_{X^n\rightarrow{Y^n}}(\cdot,\cdot)$ is often used as a measure of information from the sequence $(X^i,Y^{i-1})$ over the channel $P_{Y_i|Y^{i-1},X^{i}}(\cdot|\cdot,\cdot)$ to the random variable (RV) $Y_i$, $i=0,1,\ldots,n$. Directed information is also used in biological applications \cite{solo2008,quinn-coleman-kiyavash-hatsopoulos2011} as a measure of causality, describing the cause and effect.\\
In this paper, it is assumed that 
\begin{align}
P_{X_i|X^{i-1},Y^{i-1}}(dx_i|x^{i-1},y^{i-1})=P_{X_i|X^{i-1}}(dx_i|x^{i-1})-a.s.,~\forall~i=0,1,\ldots,n. \label{2}
\end{align}
The above assumption states that the process $\{X_i:~i=0,1,\ldots,n\}$ is conditionally independent of $Y^{i-1}=y^{i-1}$ given knowledge of $X^{i-1}=x^{i-1}$. Clearly, (\ref{2}) is implied by the following conditional independence, $P_{Y_i|Y^{i-1},X^{\infty}}$ $(dy_i|y^{i-1},x^{\infty})=P_{Y_i|Y^{i-1},X^i}(dy_i|y^{i-1},x^i)-a.s.,~\forall~i=0,1,\ldots,n$. The last assumption implies that the reproduction of $Y_i$ does not depend on future values $X_{i+1}^{\infty}\tri\{X_{i+1},X_{i+2},\ldots,X_{\infty}\}$, stating that $Y_i$ is nonanticipative with respect to the process $\{X_i:~i=0,1,\ldots,n\}$.\\
Given a sequence of source distributions $\{{P}_{X_i|X^{i-1}}(\cdot|\cdot):~i=0,1,\ldots,n\}$ and a sequence of reproduction conditional distributions $\{P_{Y_i|Y^{i-1},X^i}(\cdot|\cdot,\cdot):~i=0,1,\ldots,n\}$ define the joint distribution $P_{X^n,Y^n}(dx^n,dy^n)={P}_{X^n}(dx^n)\otimes\big(\otimes_{i=0}^n{P}_{Y_i|Y^{i-1},X^i}(dy_i|y^{i-1},x^i)\big)$. The  nonanticipative RDF is a special case of directed information defined for $i=0,1,\ldots,n$, by 
\begin{align}
I_{P_{X^n}}(X^n\rightarrow{Y^n})=\mathbb{I}_{X^n\rightarrow{Y^n}}(P_{X^n},P_{Y_i|Y^{i-1},X^i}:~i=0,1,\ldots,n)\label{3}
\end{align}
{\it Nonanticipative RDF.} The nonanticipative RDF is defined by
\begin{equation}
{R}^{na}_{0,n}(D)\tri \inf_{\substack{P_{Y_i|Y^{i-1},X^i}(\cdot|\cdot,\cdot):~i=0,1,\ldots,n:\\
\mathbb{E}\big\{d_{0,n}(X^n,Y^n)\leq{D}\big\}}}I_{P_{X^n}}(X^n\rightarrow{Y^n}).\label{7}
\end{equation}
The definition of the nonanticipative RDF is consistent with \cite{gorbunov-pinsker} in which nonanticipation is defined via the Markov chain $X_{n+1}^\infty \leftrightarrow X^n \leftrightarrow Y^n$, e.g., $P_{Y^n|X^{\infty}}(dy^n|x^{\infty})=P_{Y^n|X^n}(dy^n|x^n)$. Therefore, by finding the solution of  (\ref{7}), then one can realize it via a channel from which one can construct an optimal filter causally as in Fig.~\ref{filtering_and_causal}.
\par The paper is organized as follows. Section~\ref{abstract} discusses the formulation on abstract spaces. Section~\ref{existence} establishes  existence of optimal minimizing  distribution,  and Section~\ref{necessary} presents the optimal minimizing distribution for stationary processes, which was derived in \cite{charalambous-stavrou-ahmed2013}. Section~\ref{realization1} describes the realization of nonanticipative RDF for a vector partially observable Gaussian system over a scalar additive Gaussian noise communication channel for which the optimal causal filter is obtained. 


\section{Abstract Formulation}\label{abstract}

The source and reproduction alphabets are sequences of Polish spaces \cite{dupuis-ellis97} as defined in the previous section. Probability distributions on any measurable space  $( {\cal Z}, {\cal B}({\cal Z}))$ are denoted by ${\cal M}_1({\cal Z})$. For $({\cal X}, {\cal B}({\cal X})), ({\cal Y}, {\cal B}({\cal Y}))$  measurable spaces, the set of conditional distributions  $P_{Y|X}(\cdot|X=x)$ is denoted by ${\cal Q}({\cal Y};{\cal X})$, and these are equivalent to stochastic kernels on $({\cal Y},{\cal B}({\cal Y}))$ given $({\cal X},{\cal B}({\cal X}))$.\\
Given the process distributions $P_{X^n}(dx^n)$ and $\{P_{Y_i|Y^{i-1},X^i}(dy_i|y^{i-1},x^i):~i=0,1,\ldots,n\}$, the following probability distributions are defined.\\
({\bf P1}): The reproduction conditional probability distribution ${\overrightarrow P}_{Y^n|X^n}(dy^n|x^n)$ $ \in {\cal M}_1({\cal Y}_{0,n})$:
\begin{align}
{\overrightarrow P}_{Y^n|X^n}(dy^n|x^n)&\tri \int_{A_0}P_{Y_0|X_0}(dy_0|x_0)\int_{A_1}P_{Y_1|Y_0,X^1}(dy_1|y_0,x^1)\ldots\nonumber\\
&\ldots\int_{{A}_n}P_{Y_n|Y^{n-1},X^n}(dy_n|y^{n-1},x^n),~~A_{0,n}=\times_{i=0}^n{A_i}\in{\cal B}({\cal X}_{0,n}). \label{4}
\end{align}
({\bf P2}): The joint probability distribution $P_{X^n,Y^n}\in {\cal M}_1({\cal Y}_{0,n}\times {\cal X}_{0, n})$:
\begin{align}
P_{X^n,Y^n}(G_{0,n})&\tri(P_{X^n} \otimes \overrightarrow{P}_{Y^n|X^n})(G_{0,n}),\:G_{0,n} \in {\cal B}({\cal X}_{0,n})\times{\cal B}({\cal Y}_{0,n})\nonumber\\
&=\int \overrightarrow{P}_{Y^n|X^n}(G_{0,n,x^n}|x^n)P_{X^n}(d{x^n})\nonumber
\end{align}
where $G_{0,n,x^n}$ is the $x^n-$section of $G_{0,n}$ at point ${x^n}$ defined by $G_{0,n,x^n}\tri \{y^n \in {\cal Y}_{0,n}: (x^n, y^n) \in G_{0,n}\}$ and $\otimes$ denotes the convolution.\\
({\bf P3}): The marginal distribution $P_{Y^n}\in {\cal M}_1({\cal Y}_{0,n})$:
\begin{align}
P_{Y^n}(F_{0,n})&\tri P({\cal X}_{0, n} \times F_{0,n}),~F_{0,n} \in {\cal B}({\cal Y}_{0,n})\nonumber\\
&=\int \overrightarrow{P}_{Y^n|X^n}(({\cal X}_{0, n}\times F_{0,n})_{{x}^{n}};{x}^{n})P_{X^n}(d{x^n})\nonumber \\
&=\int \overrightarrow{P}_{Y^n|X^n}(F_{0,n}|x^n) P_{X^n}(d{x^n}).\nonumber
\end{align}
({\bf P4}): The product distribution  $\Pi_{0,n}:{\cal B}({\cal X}_{0,n}) \times
{\cal B}({\cal Y}_{0,n}) \mapsto [0,1] $ of $P_{X^n}\in{\cal M}_1({\cal X}_{0, n})$ and $P_{Y^n}\in{\cal M}_1({\cal Y}_{0, n})$:
\begin{align}
\Pi_{0,n}(G_{0,n})&\tri(P_{X^n} \times P_{Y^n})(G_{0,n}),~G_{0,n} \in {\cal B}({\cal X}_{0,n}) \times {\cal B}({\cal Y}_{0,n})\nonumber\\
&=\int_{{\cal X}_{0, n}} P_{Y^n}(G_{0,n,x^n}) P_{X^n}(dx^n).\nonumber
\end{align} 
\noi Directed information (special case) is defined via the Kullback-Leibler distance:
\begin{align}
I_{P_{X^n}}(X^n\rightarrow{Y^n})&\tri\mathbb{D}(P_{X^n,Y^n}|| \Pi_{0,n})=\mathbb{D}(P_{X^n}\otimes{\overrightarrow{P}}_{Y^n|X^n}||P_{X^n}\times{P}_{Y^n})\nonumber\\
&=\int\log \Big( \frac{d  (P_{X^n} \otimes \overrightarrow{P}_{Y^n|X^n}) }{d ( P_{X^n} \times P_{Y^n} ) }\Big) d(P_{X^n} \otimes\overrightarrow{P}_{Y^n|X^n}) \nonumber\\
& = \int \log \Big( \frac{\overrightarrow{P}_{Y^n|X^n}(dy^n|x^n)}{P_{Y^n}(dy^n)} \Big)\overrightarrow{P}_{Y^n|X^n}(dy^n|x^n)\otimes{P}_{X^n}(dx^n)\nonumber\\
&\equiv \mathbb{I}_{X^n\rightarrow{Y^n}}(P_{X^n},\overrightarrow{P}_{Y^n|X^n}).\label{re3}
 \end{align}
Note that (\ref{re3}) states that directed information is expressed as a functional of $\{P_{X^n},\overrightarrow{P}_{Y^n|X^n}\}$. \\
Define the set of all $(n+1)$-fold convolution distributions by
\begin{align*}
{\cal Q}^{c}({\cal Y}_{0,n};{\cal X}_{0,n})=\Big\{&{\overrightarrow P}_{Y^n|X^n}(dy^n|x^n)\in{\cal Q}({\cal Y}_{0,n};{\cal X}_{0,n}):\\
&{\overrightarrow P}_{Y^n|X^n}(dy^n|x^n) \tri \otimes^n_{i=0}P_{Y_i|Y^{i-1},X^i}(dy_i|y^{i-1},x^i)\Big\}.
\end{align*}
Next, the definition of nonanticipative RDF is given.
\begin{definition}\label{def1}
$(${\bf Nonanticipative Rate Distortion Function}$)$
Suppose $d_{0,n}\tri\frac{1}{n+1}\sum^n_{i=0}\rho_{0,i}(x^i,y^i)$ is ${\cal B}({\cal X}_{0,n}) \times {\cal B }( {\cal Y}_{0,n})$-measurable distortion function, and let ${\cal Q}^c_{0,n}(D)$ (assuming is non-empty) denotes the average distortion or fidelity constraint defined by
\begin{align}
 {\cal Q}^c_{{0,n}}(D)\tri\Big\{&\overrightarrow{P}_{Y^n|X^n} \in {\cal Q}^{c}({\cal Y}_{0,n};{\cal X}_{0,n}) :~\ell_{d_{0,n}}(\overrightarrow{P}_{Y^n|X^n})\tri\int d_{0,n}(x^n,y^n) \nonumber \\
&\overrightarrow{P}_{Y^n|X^n}(dy^n|x^n)\otimes{P}_{X^n}(dx^n)\leq D\Big\},~D\geq0.\label{eq2}
\end{align}
The nonanticipative RDF is defined by
\begin{align}
{R}^{na}_{0,n}(D) \tri  \inf_{{\overrightarrow{P}_{Y^n|X^n}\in{\cal Q}^c_{0,n}(D)}}{\mathbb I}_{X^n\rightarrow{Y^n}}({P}_{X^n},\overrightarrow{P}_{Y^n|X^n}).\label{ex12}
\end{align}
\end{definition}
Clearly, ${R}^{na}_{0,n}(D)$ is characterized by minimizing $\mathbb{I}_{X^n\rightarrow{Y^n}}({P}_{X^n},\overrightarrow{P}_{Y^n|X^n})$ over ${\cal Q}^c_{0,n}(D)$.


\section{Existence Of Reproduction Conditional Distribution}\label{existence}

\par In this section, the existence of the minimizing $(n+1)$-fold convolution of conditional distributions in (\ref{ex12}) is established  by using the topology of weak convergence of probability measures on Polish spaces. Before we present the relevant results we state some properties of average distortion set ${\cal Q}^c_{0,n}(D)$ and functional ${\mathbb I}_{X^n\rightarrow{Y^n}}({P}_{X^n},\overrightarrow{P}_{Y^n|X^n})$. The majority of these properties is derived in \cite{charalambous-stavrou2013} for the case of general directed information functional ${\mathbb I}_{X^n\rightarrow{Y^n}}(\overleftarrow{P}_{X^n|Y^{n-1}},\overrightarrow{P}_{Y^n|X^n})$ where
$\overleftarrow{P}_{X^n|Y^{n-1}}(\cdot|y^{n-1})=\otimes_{i=0}^n{P}_{X_i|X^{i-1},Y^{i-1}}(dx_i|x^{i-1},y^{i-1})$.
\begin{theorem}({\bf Convexity Properties})\label{convexity_properties}
Let $\{{\cal X}_n:~n\in\mathbb{N}\}$ and  $\{{\cal Y}_n:~n\in\mathbb{N}\}$ be Polish spaces. Then
\begin{description}
\item[(1)] The set $\overrightarrow{P}_{Y^n|X^n}\in{\cal Q}^{c}({\cal Y}_{0,n};{\cal X}_{0,n})$ is convex.
\item[(2)] ${\mathbb I}_{X^n\rightarrow{Y^n}}({P}_{X^n},\overrightarrow{P}_{Y^n|X^n})$ is a convex functional of $\overrightarrow{P}_{Y^n|X^n}\in{\cal Q}^{c}({\cal Y}_{0,n};{\cal X}_{0,n})$ for a fixed $P_{X^n}\in{\cal M}_1({\cal X}_{0,n})$.
\item[(3)] The set ${\cal Q}^c_{0,n}(D)$ is convex.
\end{description}
\end{theorem}
\begin{proof}
Parts (1) and (2) are derived in \cite[Theorems III.3,~III4]{charalambous-stavrou2013}. Part (3) follows from Part (1).\qed
\end{proof}
Let $BC({\cal Z})$ denotes the set of bounded continuous real-valued functions on a Polish space ${\cal Z}$. A sequence $\{P_n:n\geq1\}$ of probability measures is said to {\it converge weakly} to $P\in{\cal M}_1({\cal Z})$ if 
\begin{align}
\lim_{n\longrightarrow\infty}\int_{\cal Z}f(z)dP_n(z)=\int_{\cal Z}f(z)dP(z),~\forall{f}\in{BC}({\cal Z}).\nonumber  
\end{align}
Below, we introduce the main conditions for establishing existence of an optimal solution for the nonanticipative RDF (\ref{ex12}).   
\begin{assumption}\label{conditions-existence}
The following conditions are assumed throughout the paper.
\begin{description}
\item[(1)] ${\cal Y}_{0,n}$ is a compact Polish space, ${\cal X}_{0,n}$ is a Polish space;
\item[(2)] for all $h(\cdot){\in}BC({\cal Y}_{n})$, the function mapping
\begin{align*}
(x^{n},y^{n-1})\in{\cal X}_{0,n}\times{\cal Y}_{0,n-1}\mapsto\int_{{\cal Y}_n}h(y)P_{Y|Y^{n-1},X^n}(dy|y^{n-1},x^n)\in\mathbb{R}
\end{align*} 
is continuous jointly in the variables $(x^{n},y^{n-1})\in{\cal X}_{0,n}\times{\cal Y}_{0,n-1}$;
\item[(3)] $d_{0,n}(x^n,\cdot)$ is continuous on ${\cal Y}_{0,n}$;
\item[(4)] the distortion level $D$ is such that there exist sequence $(x^n,y^{n})\in{\cal X}_{0,n}\times{\cal Y}_{0,n}$ satisfying $d_{0,n}(x^n,y^{n})<D$.
\end{description}
\end{assumption}
Note that since ${\cal Y}_{0,n}$ is assumed to be a compact Polish space, then by \cite{dupuis-ellis97} probability measures on ${\cal Y}_{0,n}$ are weakly compact. Moreover, the following weak compactness result can be obtained, which we use to show existence of an optimal nonanticipative RDF, $R_{0,n}^{na}(D)$.
\begin{lemma}({\bf Compactness})\label{compactness2}
Suppose Assumption~\ref{conditions-existence}, (1), (2) hold.\\
Then
\begin{description}
\item[(1)] The set $\overrightarrow{P}_{Y^n|X^n}\in{\cal Q}^{c}({\cal Y}_{0,n};{\cal X}_{0,n})$ is closed and tight, hence compact.
\item[(2)] Under the additional conditions (3), (4)  the set $ {\cal Q}^c_{0,n}(D)$ is a closed subset of the compact set $\overrightarrow{P}_{Y^n|X^n}\in{\cal Q}^{c}({\cal Y}_{0,n};{\cal X}_{0,n})$, hence compact.
\end{description}
\end{lemma}
\begin{proof}
(1) The tightness of the proof is shown in from \cite[Theorem III.5, Part A., A4)]{charalambous-stavrou2013}. This follows from the fact that any ${\overrightarrow{P}}_{Y^n|X^n}(dy^n|x^n)\in{\cal Q}^{c}({\cal Y}_{0,n};{\cal X}_{0,n})$ is factorized as ${\overrightarrow{P}}_{Y^n|X^n}(dy^n|x^n)=\otimes_{i=0}^n{P}_{Y_i|Y^{i-1},X^i}(dy_i|y^{i-1},x^i)$-a.s., where ${P}_{Y_i|Y^{i-1},X^i}(dy_i|y^{i-1},x^i)\in{\cal Q}({\cal Y}_i;{\cal Y}_{0,i-1}\times{\cal X}_{0,i})\subset{\cal M}_1({\cal Y}_i)$, $i=0,1,\ldots,{n}$, and ${\cal Y}_{0,n}$ compact Polish space which implies that $\{P_{Y_i|Y^{i-1},X^i}(\cdot|y^{i-1},x^i):y_0\in{\cal Y}_{0},y_1\in{\cal Y}_{1},\ldots,y_{i-1}\in{\cal Y}_{i-1},x^i\in{\cal X}_{0,i}\}$ is compact, hence by Prohorov's theorem it is uniformly tight $\forall{i}$.\\
Therefore, by Prohorov's theorem \cite{dupuis-ellis97} the compactness of the set $\overrightarrow{P}_{Y^n|X^n}\in{\cal Q}^{c}({\cal Y}_{0,n};{\cal X}_{0,n})$ will follow if we show that it is closed, i.e., given $\{\overrightarrow{P}_{Y^n|X^n}^{\alpha}(\cdot|x^{n}):\alpha=1,2,\ldots\}\subset{\cal Q}^{c}({\cal Y}_{0,n};{\cal X}_{0,n})$ with $\overrightarrow{P}_{Y^n|X^n}^{\alpha}(\cdot|x^{n}) \buildrel w \over \longrightarrow\overrightarrow{P}_{Y^n|X^n}^{0}(\cdot|x^{n})$ then $\overrightarrow{P}_{Y^n|X^n}^{0}(\cdot|x^{n})\in{\cal Q}^{c}({\cal Y}_{0,n};{\cal X}_{0,n})$. Since the family of measures $\overrightarrow{P}_{Y^n|X^n}(\cdot|x^{n})\in{\cal Q}^{c}({\cal Y}_{0,n};{\cal X}_{0,n})$ and $\{P_{Y_i|Y^{i-1},X^i}(\cdot;y^{i-1},x^{i}):~i=0,1,\ldots,n\}$, are tight, and $P_{Y_i|Y^{i-1},X^i}(\cdot;y^{i-1},x^{i})$ are probability measures on ${\cal M}_1({\cal Y}_i)$, $i=0,1,\ldots,n$, then, for $\overrightarrow{P}_{Y^n|X^n}^{\alpha}(\cdot|x^{n})\in{\cal Q}^{c}({\cal Y}_{0,n};{\cal X}_{0,n})$ there is a collection of probability measures $\{P_{Y_i|Y^{i-1},X^i}(\cdot;y^{i-1},x^{i}):i=0,1,\ldots,n\}$ such that
\begin{align*}
P^{\alpha}_{Y_i|Y^{i-1},X^i}(\cdot;y^{i-1},x^{i})\buildrel w \over \longrightarrow{P}^{0}_{Y_i|Y^{i-1},X^i}(\cdot;y^{i-1},x^{i}),~i=0,1,\ldots,n.
\end{align*}
Hence, to show closedness of $\overrightarrow{P}_{Y^n|X^n}\in{\cal Q}^{c}({\cal Y}_{0,n};{\cal X}_{0,n})$ it suffices to show that
\begin{align*}
\otimes_{i=0}^n{P}^{\alpha}_{Y_i|Y^{i-1},X^i}(\cdot;y^{i-1},x^{i})\buildrel w \over \longrightarrow\otimes_{i=0}^n{P}^0_{Y_i|Y^{i-1},X^i}(\cdot;y^{i-1},x^{i})
\end{align*}
whenever $P^{\alpha}_{Y_i|Y^{i-1},X^i}(\cdot;y^{i-1},x^{i})\buildrel w \over \longrightarrow{P}^{0}_{Y_i|Y^{i-1},X^i}(\cdot;y^{i-1},x^{i}),~i=0,1,\ldots,n$. Utilizing Assumptions~\ref{conditions-existence} (2), this can be shown by induction, and hence $\overrightarrow{P}_{Y^n|X^n}\in{\cal Q}^{c}({\cal Y}_{0,n};{\cal X}_{0,n})$ is also a closed set. This completes the derivation of part (1).\\
(2) Utilizing compactness of $\overrightarrow{P}_{Y^n|X^n}\in{\cal Q}^{c}({\cal Y}_{0,n};{\cal X}_{0,n})$, condition (3) of Assumption~\ref{conditions-existence} on  $d_{0,n}(x^n,\cdot)$, and some fundamental measure theoretic results like Lebesgue's monotone convergence theorem and Fatou's lemma, it can be shown that $ {\cal Q}^c_{0,n}(D)$ is a closed subset of ${\cal Q}^{c}({\cal Y}_{0,n};{\cal X}_{0,n})$, and hence by Prohorov's theorem it is compact.\qed
\end{proof}

\noi The previous results utilize Prohorov's theorem that relates tightness and weak compactness.
\par The next theorem establishes existence of the minimizing reproduction distribution for (\ref{ex12}). We need the following theorem derived in \cite{charalambous-stavrou2013}.
\begin{lemma}({\bf Lower Semicontinuity})\label{lower-semicontinuity}
Under Assumption~\ref{conditions-existence} (1), (2), ${\mathbb{I}}_{X^n\rightarrow{Y^n}}({P}_{X^n}, {\overrightarrow P}_{Y^n|X^n})$ is lower semicontinuous on ${\overrightarrow P}_{Y^n|X^n}\in{\cal Q}^{c}({\cal Y}_{0,n};{\cal X}_{0,n})$ for a fixed ${P}_{X^n}\in{\cal M}_1({\cal X}_{0,n})$.
\end{lemma}
\begin{proof}
The proof is immediate from \cite[Theorem III.7]{charalambous-stavrou2013}, and it is obtained by just relegating the general directed information functional ${\mathbb I}_{X^n\rightarrow{Y^n}}(\overleftarrow{P}_{X^n|Y^{n-1}},\overrightarrow{P}_{Y^n|X^n})$  to the special case of ${\mathbb{I}}_{X^n\rightarrow{Y^n}}({P}_{X^n}, {\overrightarrow P}_{Y^n|X^n})$ where
$\overleftarrow{P}_{X^n|Y^{n-1}}(\cdot|y^{n-1})=P_{X^n}(x^{n})-a.s.$\qed
\end{proof}
By Lemma~\ref{compactness2} and Lemma~\ref{lower-semicontinuity} we have the following existence result.
\begin{theorem}({\bf Existence})\label{existence_rd}
Suppose the conditions and results of Lemma~\ref{compactness2} and Lemma~\ref{lower-semicontinuity} hold. Then ${R}^{na}_{0,n}(D)$ has a minimum.
\end{theorem}
\begin{proof}
Provided that the results from Lemma~\ref{compactness2} and Lemma~\ref{lower-semicontinuity} hold, then the existence of a global minimum solution follows from an extended version of Weierstrass' theorem (e.g., a lower semicontinuous function on a compact set attains its minimum).\qed
\end{proof} 
\noi The fundamental difference between \cite{charalambous-stavrou-ahmed2013} and this paper, is that we show existence of solution to the nonanticipative RDF using the topology of weak convergence of probability measures under very relaxed conditions, i.e., Assumption~\ref{conditions-existence}. These are natural generalization of the existence result discussed by Csisz\'ar in \cite{csiszar74}, for single letter classical RDF.


\section{Optimal Reproduction of Nonanticipative RDF for Stationary Processes}\label{necessary}

In this section, we present the form of the optimal stationary reproduction conditional distribution. Since we have shown existence of solution to the nonanticipative RDF, the method of obtaining the optimal solution is identical to the one in \cite[Section IV]{charalambous-stavrou-ahmed2013}. We introduce the following main assumption.
\begin{assumption}$(${\bf Stationarity}$)$\label{stationarity}
The $(n+1)$-fold convolution of conditional distribution $\overrightarrow{P}_{Y^n|X^n}(dy^n|x^n)=\otimes^n_{i=0}P_{Y_i|Y^{i-1},X^i}$ $(dy_i|y^{i-1},x^i)$, is the convolution of stationary conditional distributions.
\end{assumption}

\noi Assumption~\ref{stationarity} holds for stationary process $\{(X_i,Y_i):i\in\mathbb{N}\}$ and $\rho_{0,i}(x^i,y^i)\equiv\rho(T^i{x^n},T^i{y^n})$, where $T^i{x^n}=\tilde{x}^{n}$ is the $i^{th}$ shift operator on the source sequence $x^n$, with $\tilde{x}_{n}=x_{n+i}$ (similarly for $T^i{y^n}$), and $\sum_{i=0}^n\rho(T^ix^n,T^iy^n)$ depends only on the components of $(x^n,y^n)$  \cite{gray2010}. The consequence of Assumption~\ref{stationarity}, which holds for stationary processes  and a single letter distortion function, is that the Gateaux differential of $\mathbb{I}_{X^n\rightarrow{Y^n}}(P_{X^n},\overrightarrow{P}_{Y^n|X^n})$ is done in only one direction $\big{(}$since $P_{Y_i|Y^{i-1},X^i}(dy_i|y^{i-1},x^i)$ are stationary$\big{)}$. Therefore, we define the variation of $\overrightarrow{P}_{Y^n|X^n}$ in the direction of $\overrightarrow{P}_{Y^n|X^n}-\overrightarrow{P}_{Y^n|X^n}^0$ via $\overrightarrow{P}_{Y^n|X^n}^{\epsilon}\tri\overrightarrow{P}_{Y^n|X^n}+\epsilon\big{(}\overrightarrow{P}_{Y^n|X^n}-\overrightarrow{P}_{Y^n|X^n}^0\big{)}$, $\epsilon\in[0,1]$, since under Assumption~\ref{stationarity}, the functionals $\{P_{Y_i|Y^{i-1},X^i}(dy_i|y^{i-1},x^i)\in{\cal Q}({\cal Y}_i;{\cal Y}_{0,i-1}\times{\cal X}_{0,i}):~i=0,1,\ldots,n\}$ are identical.
\begin{theorem} \label{th5}
Suppose Assumption~\ref{stationarity} holds and~${\mathbb I}_{P_{X^n}}(\overrightarrow{P}_{Y^n|X^n}) \tri\mathbb{I}_{X^n\rightarrow{Y^n}}$ $(P_{X^n},\overrightarrow{P}_{Y^n|X^n})$ is well defined for every $\overrightarrow{P}_{Y^n|X^n}\in {\cal Q}^c_{0,n}(D)$ possibly taking values from the set $[0,\infty]$. Then  $\overrightarrow{P}_{Y^n|X^n} \rightarrow {\mathbb I}_{P_{X^n}}(\overrightarrow{P}_{Y^n|X^n})$ is Gateaux differentiable at every point in ${\cal Q}^c_{0,n}(D)$, and the Gateaux derivative at the  point $\overrightarrow{P}_{Y^n|X^n}^0$ in the direction $\overrightarrow{P}_{Y^n|X^n}-\overrightarrow{P}_{Y^n|X^n}^0$ is given
by
\begin{eqnarray}
&&\delta{\mathbb I}_{P_{X^n}}(\overrightarrow{P}_{Y^n|X^n}^0,\overrightarrow{P}_{Y^n|X^n}-\overrightarrow{P}_{Y^n|X^n}^0)\nonumber\\
&&=\int\log \Bigg(\frac{\overrightarrow{P}_{Y^n|X^n}^0(dy^n|x^n)}{P_{Y^n}^0(dy^n)}\Bigg)(\overrightarrow{P}_{Y^n|X^n}-\overrightarrow{P}_{Y^n|X^n}^0)(dy^n|x^n) P_{X^n}(dx^n)\nonumber
\end{eqnarray}
where $P_{Y^n}^0\in{\cal M}_1({\cal Y}_{0,n})$ is the marginal measure corresponding
to $\overrightarrow{P}_{Y^n|X^n}^0\otimes{P}_{X^n}(dx^n)\in{\cal M}_1({\cal Y}_{0,n}\times{\cal X}_{0,n})$.
\end{theorem}
\begin{proof}
The proof is similar to the one in \cite[Theorem 4.1]{farzad06}.\qed
\end{proof}
\noi By Theorem~\ref{convexity_properties}, the nonanticipative RDF is a convex optimization problem, and by Theorem~\ref{existence_rd} a solution exists. Hence, using these theorems, it can be shown that the constrained problem defined by (\ref{ex12}) can be reformulated as an unconstrained problem using Lagrange multipliers. This procedure is described in \cite[Theorem IV.3]{charalambous-stavrou-ahmed2013}, hence it is omitted; preferably, we state the main result, that is, the optimal reproduction conditional distribution that characterized nonanticipative RDF is defined.
\begin{theorem}({\bf Optimal Reproduction of Nonanticipative RDF}) \label{th6}
Suppose the Assumptions~\ref{conditions-existence} hold and consider $d_{0,n}(x^n,y^n)\tri\frac{1}{n+1}\sum_{i=0}^n\rho(T^i{x^n},T^i{y^n})$. Then
\begin{description}
\item[(1)] The infimum is attained at  $\overrightarrow{P}^*_{Y^n|X^n} \in{\cal Q}^c_{0,n}(D)$ given by\footnote{Due to stationarity assumption $P_{Y_i|Y^{i-1}}(\cdot|\cdot)=P(\cdot|\cdot)$ and ${P}^*_{Y_i|Y^{i-1},X^i}(\cdot|\cdot,\cdot)={P}^*(\cdot|\cdot,\cdot)$}
\begin{align}
\overrightarrow{P}^*_{Y^n|X^n}(dy^n|x^n)&=\otimes_{i=0}^n{P}^*_{Y_i|Y^{i-1},X^i}(dy_i|y^{i-1},x^i)\nonumber\\
&=\otimes_{i=0}^n\frac{e^{s \rho(T^i{x^n},T^i{y^n})}P^*_{Y_i|Y^{i-1}}(dy_i|y^{i-1})}{\int_{{\cal Y}_i} e^{s \rho(T^i{x^n},T^i{y^n})} P^*_{Y_i|Y^{i-1}}(dy_i|y^{i-1})},~s\leq{0}\label{ex14}
\end{align}
and $P^*_{Y_i|Y^{i-1}}(dy_i|y^{i-1})\in {\cal Q}({\cal Y}_i;{\cal Y}_{0,{i-1}})$. 
\item[(2)] The nonanticipative RDF is given by
\begin{align}
{R}_{0,n}^{na}(D)&\left.=(n+1)sD -\sum_{i=0}^n\int\log \Big( \int_{{\cal Y}_i} e^{s\rho(T^i{x^n},T^i{y^n})} P^*_{Y_i|Y^{i-1}}(dy_i|y^{i-1})\Big)\right.\nonumber\\[-1.5ex]\label{ex15}\\[-1.5ex]
&\quad\left.\times{\overrightarrow{P}^*_{Y^{i-1}|X^{i-1}}(dy^{i-1}|x^{i-1})\otimes{P}_{X^i}(dx^i).}\right.\nonumber
\end{align}
If ${R}_{0,n}^{na}(D) > 0$ then $ s < 0$  and
\begin{eqnarray}
\frac{1}{n+1}\sum_{i=0}^n\int\rho(T^i{x^n},T^i{y^n})\overrightarrow{P}^*_{Y^{i}|X^{i}}(dy^i;x^i)P_{X^i}(dx^i)=D.\label{eq.7}
\end{eqnarray}
\end{description}
\end{theorem}
\begin{proof}
The proof is similar to \cite[Theorem IV.4]{charalambous-stavrou-ahmed2013}, hence it is omitted.\qed
\end{proof}

\begin{remark}
Note that if the distortion function satisfies $\rho(T^i{x^n},T^i{y^n})=\rho(x_i,T^i{y^n})$ then
\begin{align*}
{P}^*_{Y_i|Y^{i-1},X^i}(dy_i|y^{i-1},x^i)={P}^*_{Y_i|Y^{i-1},X^i}(dy_i|y^{i-1},x_i)-a.s.,~i=0,1,\ldots,n.
\end{align*}
That is, the reproduction kernel is Markov in $X^n$. However, without further restrictions one cannot claim that this conditional distribution is also Markov with respect to $\{Y_i:~i=0,1,\ldots,n\}$.
\end{remark}
Note that unlike \cite{charalambous-stavrou-ahmed2013}, we have derived the main results using the topology of weak convergence of probability measures.

\section{Realization Of Nonanticipative Rate Distortion Function and Example}\label{realization1}

\par In this section, we first describe the construction of the filter using the optimal solution of the nonanticipative RDF and then we present the multidimensional partially observable Gaussian system, which is realizable over a scalar additive Gaussian noise channel.

\subsection{Realization of the Nonanticipative RDF}
The realization of the nonanticipative RDF (optimal reproduction conditional distribution) is equivalent to the sensor mapping as shown in Fig.~\ref{filtering_and_causal}, which produces the auxiliary random process $\{Z_i:~i\in\mathbb{N}\}$ that will be used for filtering. This is equivalent to identifying a communication channel, an encoder and a decoder such that the reproduction from the sequence $X^n$ to the sequence $Y^n$ matches the nonanticipative rate distortion minimizing reproduction kernel. Fig.~\ref{realization_figure} illustrates the cascade  sub-systems that realize the nonanticipative RDF, which is consistent with the discussion in the introduction.
\begin{definition}$(${\bf Realization}$)$\label{realization}
Given a source $\{P_{X_i|X^{i-1}}(dx_i|x^{i-1}):i=0,\ldots,n\}$,  a channel $\{P_{B_i|B^{i-1},A^{i}}(db_i|b^{i-1},a^i):i=0,\ldots,n\}$ is a realization of the optimal reproduction distribution if there exists a pre-channel encoder $\{P_{A_i|A^{i-1},B^{i-1},X^i}$ $(da_i|a^{i-1},b^{i-1},x^i):i=0,\ldots,n\}$ and a post-channel decoder $\{P_{Y_i|Y^{i-1},B^i}(dy_i|$ $y^{i-1},b^i):i=0,\ldots,n\}$ such that
\begin{align}
 {\overrightarrow {P}}_{Y^{n}|X^{n}}^*(dy^n|x^n)\tri\otimes_{i=0}^n P^*_{Y_i|Y^{i-1},X^i}(dy_i|y^{i-1},x^i)\nonumber
\end{align}
where the joint distribution is
\begin{align}
&P_{X^n, A^n, B^n, Y^n}(dx^n,da^n,db^n,dy^n) \nonumber \\
&=\otimes_{i=0}^n P_{Y_i|Y^{i-1},B^i}(dy_i|y^{i-1},b^i)  \otimes P_{B_i|B^{i-1},A^{i}}
(db_i|b^{i-1},a^i) \nonumber \\
&~~~\otimes P_{A_i|A^{i-1},B^{i-1},X^i}(da_i|a^{i-1},b^{i-1},x^i)\otimes P_{X_i|X^{i-1}}(dx_i|x^{i-1}) \nonumber
\end{align}
The filter is given by $\{P_{X_i|B^{i-1}}(dx_i|b^{i-1}):i=0,\ldots,n\}$ or by $\{P_{X_i|Y^{i-1}}(dx_i|y^{i-1}):i=0,\ldots,n\}$.
\end{definition}
\begin{figure}[H]
\centering
\includegraphics[scale=0.75]{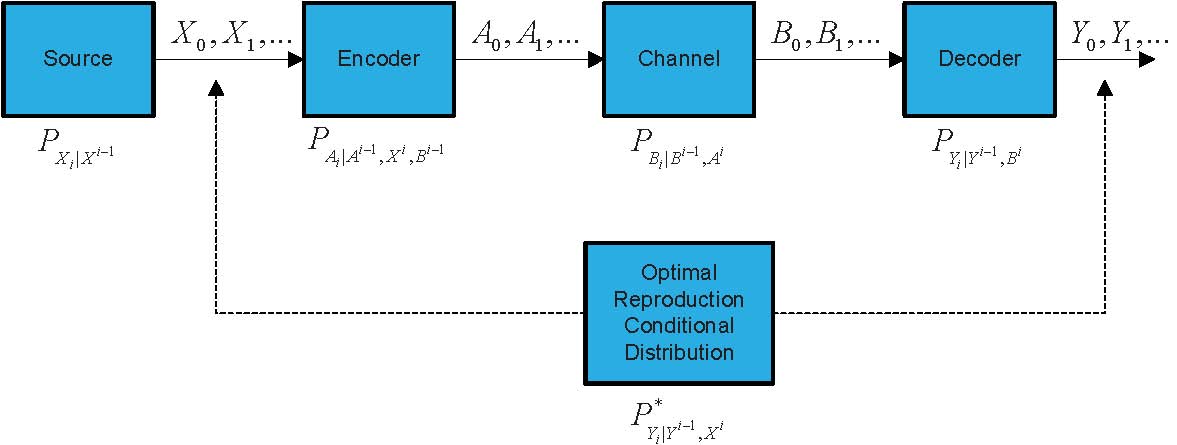}
\caption{Block Diagram of Realizable Nonanticipative Rate Distortion Function}
\label{realization_figure}
\end{figure}
\noi Thus, if $\{P_{B_i|B^{i-1},A^{i}}(db_i|b^{i-1},a^i):i=0,\ldots,n\}$ is a realization of the nonanticipative RDF minimizing distribution  then the channel connecting the source, encoder, channel, decoder achieves the nonanticipative RDF, and the filter is obtained. Clearly, $\{B_i:~i=0,1,\ldots,n\}$ is an auxiliary random process which is needed to obtain the filter $\{P_{X_i|B^{i-1}}(dx_i|b^{i-1}):i=0,\ldots,n\}$.\\
Note that if we also impose the requirement that the channel capacity is equal to the $\lim_{n\longrightarrow\infty}\frac{1}{n+1}R^{na}_{0,n}(D)$, then the realization procedure described in Definition~\ref{realization}, is equal to the joint source-channel matching \cite{gastpar2003} for sources with memory without anticipation. Next, we present an example where the optimal communication via symbol-by-symbol or uncoded transmission is established.


\subsection{Example: Mutlidimensional Gaussian Processes over a Scalar AWGN Channel}\label{example}

\par In this section, we present an example to illustrate the realization procedure described in Section~\ref{realization1}. We consider multidimensional Gaussian-Markov sources transmitted optimally over a scalar additive Gaussian channel.  Hence, this example is distinguished from a companion example described in \cite[Section VI]{charalambous-stavrou-ahmed2013}  where the multidimensional Gaussian-Markov source is transmitted optimally over a vector additive Gaussian channel.\\
Consider the following discrete-time partially observed linear Gauss-Markov system described by
\begin{eqnarray}
\left\{ \begin{array}{ll} X_{t+1}=AX_t+BW_t,~X_0=X\in\mathbb{R}^n,~t\in\mathbb{N}\\
Y_t=CX_t+GV_t,~t\in\mathbb{N} \end{array} \right.\label{equation51}
\end{eqnarray}
\noi where $X_t\in\mathbb{R}^m$ is the state (unobserved) process of information source (plant), and $Y_t\in\mathbb{R}^p$ is the partially measurement (observed) process. Assume that ($C,A$) is detectable and ($A,\sqrt{BB^{tr}}$) is stabilizable, ($G\neq0$). The state and observation noises $\{(W_t,V_t):t\in\mathbb{N}\}$, $W_t\in\mathbb{R}^k$ and $V_t\in\mathbb{R}^p$, are Gaussian IID processes with zero mean and identity covariances are mutually independent, and independent of the Gaussian RV $X_0$, with parameters $N(\bar{x}_0,\bar{V}_0)$.

\noi The objective is to reconstruct $\{Y_t:~t\in\mathbb{N}\}$ from $\{\tilde{Y}_t:~t\in\mathbb{N}\}$ causally. The distortion is single letter defined by 
\begin{eqnarray*}
d_{0,n}(y^n,\tilde{y}^n)\tri\frac{1}{n+1}\sum_{t=0}^n||y_t-\tilde{y}_t||_{\mathbb{R}^p}^2.
\end{eqnarray*}
The objective is to compute 
\begin{eqnarray*}
R_{0,n}^{na}(D)=\inf_{\overrightarrow{P}_{\tilde{Y}^n|Y^n}\in{\cal Q}^c_{0,n}(D)}\frac{1}{n+1}\mathbb{I}_{Y^n\rightarrow\tilde{Y}^n}(P_{Y^n},\overrightarrow{P}_{\tilde{Y}^n|Y^n})
\end{eqnarray*}
and then realize the reproduction distribution. According to Theorem~\ref{th6}, the optimal reproduction is given by 
\begin{eqnarray}
\overrightarrow{P}^*_{\tilde{Y}^n|Y^n}(d\tilde{y}^n|y^n)=\otimes_{t=0}^n\frac{e^{s||\tilde{y}_t-y_t||_{\mathbb{R}^p}^2}P_{\tilde{Y}_t|\tilde{Y}^{t-1}}(d\tilde{y}_t|\tilde{y}^{t-1})}{\int_{\tilde{\cal Y}_t}e^{s||\tilde{y}_t-y_t||_{\mathbb{R}^p}^2}P_{\tilde{Y}_t|\tilde{Y}^{t-1}}(d\tilde{y}_t|\tilde{y}^{t-1})},~s\leq{0}.\label{eq.9}
\end{eqnarray}
Hence, from (\ref{eq.9}) it follows that $P_{\tilde{Y}_t|\tilde{Y}^{t-1},Y^t}=P_{\tilde{Y}_t|\tilde{Y}^{t-1},Y_t}(d\tilde{y}_t|\tilde{y}^{t-1},y_t)-a.a.~(\tilde{y}^{t-1},y_t)$, that is the reproduction is Markov with respect to the process $\{Y_t:~t\in\mathbb{N}\}$. Moreover, since the exponential term $||\tilde{y}_t-y_t||_{\mathbb{R}^p}^2$ in the right hand side of (\ref{eq.9}) is quadratic in $(y_t,\tilde{y}_t)$, and $\{X_t:~t\in\mathbb{N}\}$ is Gaussian, then $\{(X_t,{Y}_t):~t\in\mathbb{N}\}$ is jointly Gaussian, hence it follows that $P_{\tilde{Y}_t|\tilde{Y}^{t-1},Y_t}(\cdot|\tilde{y}^{t-1},y_t)$ is Gaussian (for a fixed realization of $(\tilde{y}^{t-1},y_t)$). Hence, it has the general form
\begin{eqnarray}
\tilde{Y}_t=\bar{A}_tY_t+\bar{B}_t\tilde{Y}^{t-1}+\bar{Z}_t,~t\in\mathbb{N}\label{eq.10}
\end{eqnarray}
where $\bar{A}_t\in\mathbb{R}^{p\times{p}}$, $\bar{B}_t\in\mathbb{R}^{p\times{t}p}$, and $\{Z_t:~t\in\mathbb{N}\}$ is an independent sequence of Gaussian vectors.\\
Next, we chose to realize (\ref{eq.10}) over a scalar additive Gaussian noise channel with feedback defined by 
\begin{eqnarray}
B_t=A_t+Z_t,~t\in\mathbb{N}\label{eq.11}
\end{eqnarray}
where the encoder is a mapping $A_t=\Phi_t(Y_t,\tilde{Y}^{t-1})$ with power $P_t\tri\mathbb{E}\{(A_t)^2\}$ as shown in Fig.~\ref{discrete_time_communication_system}. 
\begin{figure}[ht]
\centering
\includegraphics[scale=0.75]{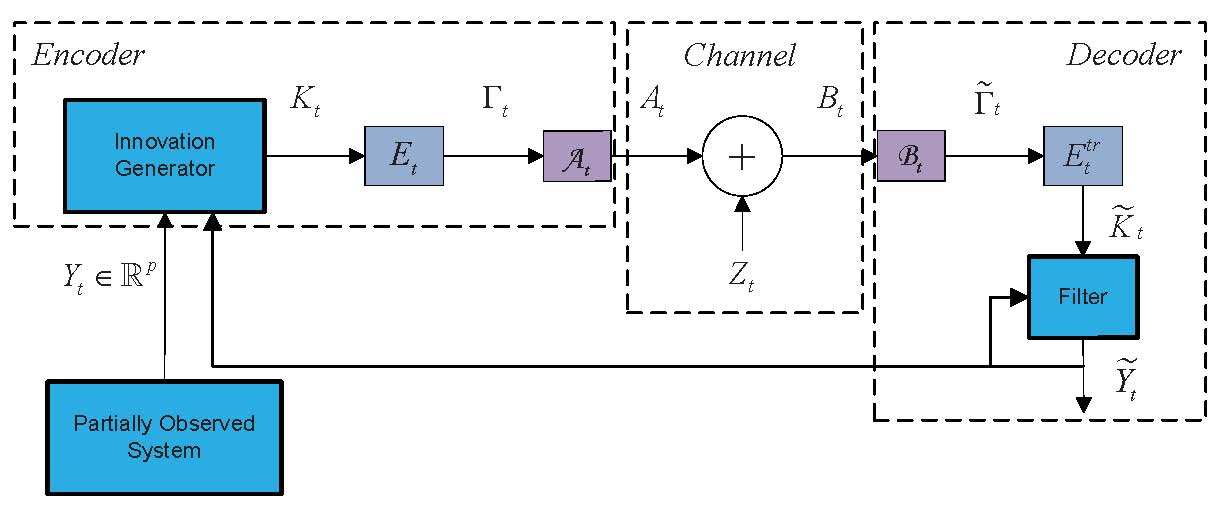}
\caption{Design of Realizable Nonanticipative Rate Distortion Function}
\label{discrete_time_communication_system}
\end{figure}
\noi Recall that for $A_t$ Gaussian, the directed information is $I(A^t\rightarrow{B}^t)=\log\big(1+\mathbb{E}\{(A_t)^2\}Var(Z_t)^{-1}\big)$. The decoder at time $t\in\mathbb{N}$ receives $B^t$ and computes the reproduction $\tilde{Y}_t=\Psi_t(B^t,\tilde{Y}^{t-1})$.\\
{\it Calculation of Nonanticipative RDF.} First, we compute the Gaussian innovation process $\{K_t:~t\in\mathbb{N}\}$, defined by 
\begin{eqnarray}
K_t\tri{Y}_t-\mathbb{E}\Big{\{}Y_t|\sigma\{\tilde{Y}^{t-1}\}\Big{\}},~t\in\mathbb{N}\label{equation52}
\end{eqnarray}
whose covariance is defined by $\Lambda_t\tri{E}\{K_tK_t^{tr}\}$. The decoder consists of a pre-decoder $\{\tilde{K}_t:~t\in\mathbb{N}\}$ which is defined by 
\begin{eqnarray}
\tilde{K}_t\tri\tilde{Y}_t-\mathbb{E}\Big{\{}Y_t|\sigma\{\tilde{Y}^{t-1}\}\Big{\}},~t\in\mathbb{N}.\label{eq.12}
\end{eqnarray}
Note that the fidelity criterion satisfies $d_{0,n}(Y^n,\tilde{Y}^n)=d_{0,n}(K^n,\tilde{K}^n)=\frac{1}{n+1}\sum_{t=0}^n||\tilde{K}_t-K_t||_{\mathbb{R}^p}^2=\frac{1}{n+1}\sum_{t=0}^n||\tilde{\Gamma}_t-\Gamma_t||_{\mathbb{R}^p}^2$. Let $\{E_t:~t\in\mathbb{N}\}$ be the unitary matrix that diagonalizes $\{\Lambda_t:~t\in\mathbb{N}\}$, such that
\begin{eqnarray}
E_t\Lambda_t{E}_t^{tr}=diag\{\lambda_{t,1},\ldots\lambda_{t,p}\},~t\in\mathbb{N}.\label{equation53}
\end{eqnarray}
\noi Define $\Gamma_t\tri{E}_tK_t$. Then $\{\Gamma_t:~t\in\mathbb{N}\}$ has independent components. Let $\{\tilde{\Gamma}_t:~t\in\mathbb{N}\}$, where $\tilde{\Gamma}=E_t\tilde{K}_t$ denote its reproduction and define $d_{0,n}(\Gamma^n,\tilde{\Gamma}^n)\tri\frac{1}{n+1}\sum_{t=0}^n||\Gamma_t-\tilde{\Gamma}_t||_{\mathbb{R}^p}^2$. Then by \cite{blahut1987} (invoking an upper and Shannon's lower bound if necessary) we can have,
\begin{align}
R^{na}(D)&=\lim_{n\longrightarrow\infty}R_{0,n}^{na}(D)\tri\lim_{n\longrightarrow\infty}\inf_{\substack{\overrightarrow{P}_{\tilde{Y}^n|Y^n}:\\~\mathbb{E}\big{\{}d_{0,n}(Y^n,\tilde{Y}^n)\leq{D}\big{\}}}}\frac{1}{n+1}\mathbb{I}_{Y^n\rightarrow\tilde{Y}^n}(P_{Y^n},\overrightarrow{P}_{\tilde{Y}^n|Y^n})\nonumber\\
&=\lim_{n\longrightarrow\infty}R_{0,n}^{na,K^n,\tilde{K}^n}(D)\tri\lim_{n\longrightarrow\infty}\inf_{\substack{\overrightarrow{P}_{\tilde{K}^n|K^n}:\\~\mathbb{E}\big{\{}d_{0,n}(K^n,\tilde{K}^n)\leq{D}\big{\}}}}\frac{1}{n+1}\mathbb{I}(P_{K^n},\overrightarrow{P}_{\tilde{K}^n|K^n})\nonumber\\
&=\lim_{n\longrightarrow\infty}R_{0,n}^{na,\Gamma^n,\tilde{\Gamma}^n}(D)\tri\lim_{n\longrightarrow\infty}\inf_{\substack{\overrightarrow{P}_{\tilde{\Gamma}^n|\Gamma^n}:\\~\mathbb{E}\big{\{}d_{0,n}(\Gamma^n,\tilde{\Gamma}^n)\leq{D}\big{\}}}}\frac{1}{n+1}\mathbb{I}(P_{\Gamma^n},\overrightarrow{P}_{\tilde{\Gamma}^n|\Gamma^n})\\
&=\lim_{n\longrightarrow\infty}\frac{1}{2}\frac{1}{n+1}\sum_{t=0}^n\sum_{i=1}^p\log\Big{(}\frac{\lambda_{t,i}}{\delta_{t,i}}\Big{)}=\frac{1}{2}\sum_{i=1}^p\log\Big{(}\frac{\lambda_{\infty,i}}{\delta_{\infty,i}}\Big{)}\nonumber
\end{align}
where $\lambda_{\infty,i}\tri\lim_{t\longrightarrow\infty}\lambda_{t,i}$, $\delta_{\infty,i}\tri\lim_{t\longrightarrow\infty}\delta_{t,i}$, $\xi_{\infty}\tri\lim_{t\longrightarrow\infty}\xi_t$ and 
\begin{eqnarray}
\delta_{t,i} \tri \left\{ \begin{array}{ll} \xi_t & \mbox{if} \quad \xi_t\leq\lambda_{t,i} \\
\lambda_{t,i} &  \mbox{if}\quad\xi_t>\lambda_{t,i} \end{array} \right.,~t\in\mathbb{N},~i=1,\ldots,p\nonumber
\end{eqnarray}
and $\{\xi_t:~t\in\mathbb{N}\}$ satisfies $\sum_{i=1}^p\delta_{t,i}=D$.\\
Define 
\begin{align*} 
H_{\infty}&\tri\lim_{t\longrightarrow\infty}{H}_t,{H}_t\tri{diag}\{\eta_{t,1},\ldots,\eta_{t,p}\}\in\mathbb{R}^{p\times{p}}, 
\eta_{t,i}\tri{1}-\frac{\delta_{t,i}}{\lambda_{t,i}},~t\in\mathbb{N},~i=1,\ldots,p\\
\Delta_{\infty}&\tri\lim_{t\longrightarrow\infty}\Delta_t, \Delta_t\tri{diag}\{\delta_{t,1},\ldots,\delta_{t,p}\},~t\in\mathbb{N},~i=1,\ldots,p.
\end{align*}
The reproduction conditional distributions is given by
\begin{eqnarray}
\overrightarrow{P}^*_{\tilde{\Gamma}^n|\Gamma^n}(d\tilde{\gamma}^n|{\gamma}^n)=\otimes_{t=0}^n{P}^*_{\Gamma_t|\tilde{\Gamma}_t}(d\tilde{\gamma}_t|\gamma_t),~~{P}^*_{\Gamma_t|\tilde{\Gamma}_t}(\cdot|\cdot)\sim{N}(H_t\Gamma_t,H_t\Delta_t).\nonumber
\end{eqnarray}
\noi Hence, from Fig.~\ref{discrete_time_communication_system} the reproduction is obtained from
\begin{align*}
\tilde{Y}_t&=\tilde{K}_t+\mathbb{E}\big\{Y_t|\sigma\{\tilde{Y}^{t-1}\}\big\}=E_t^{tr}\tilde{\Gamma}_t+\mathbb{E}\big\{Y_t|\sigma\{\tilde{Y}^{t-1}\}\big\}\\
&=E_t^{tr}H_tE_tK_t+E_t^{tr}\sqrt{H_t\Delta_t}Z_t+\mathbb{E}\big\{Y_t|\sigma\{\tilde{Y}^{t-1}\}\\
&=E_t^{tr}H_tE_tC\big(X_t-\mathbb{E}\big\{X_t|\sigma\{\tilde{Y}^{t-1}\}\big\}\big)+E_t^{tr}H_tE_tDV_t+E_t^{tr}\sqrt{H_t\Delta_t}Z_t+\mathbb{E}\big\{Y_t|\sigma\{\tilde{Y}^{t-1}\}\big\}.
\end{align*}
The received signal is decompressed by $\tilde{\Gamma}_t={\cal B}_tB_t$ at the pre-decoder. By the knowledge of the decoder output $\tilde{Y}^{t-1}$, the mean square estimator $\hat{X}_t$ is generated at the decoder (and encoder because $\widehat{X}_{t|t-1}\tri\mathbb{E}\big{\{}X_t|\sigma\{\tilde{Y}_{t-1}\}\big{\}}$). Next we pick a specific AWGN channel and we show how to realize the reproduction distribution, see Fig.~\ref{discrete_time_communication_system}.\\
\noi{\it Realization over a Scalar AWGN Channel.} Consider a scalar channel $B_t=A_t+Z_t,~t\in\mathbb{N}$, where $Z_t$ is Gaussian zero mean, $Q\tri{Var}(Z_t)$, and $A_t\in\mathbb{R}$. Since by data processing inequality $I(X^n\rightarrow{Y^n})\geq{I}(A^n\rightarrow{Y^n})$, then we should compress the information signal $\{\Gamma_t:t\in\mathbb{N}\}$ before we send it over the AWGN channel. Thus, we define
\begin{align*}
A_t={\cal A}_t\Gamma_t={\cal A}_tE_tK_t,~{\cal A}_t\in\mathbb{R}^{1\times{p}}.
\end{align*}
Since the channel capacity and nonanticipative RDF must be equal we set
\begin{align*}
C(P)\tri\lim_{n\longrightarrow\infty}C_{0,n}(P_0,\ldots,P_n)&=\lim_{n\longrightarrow\infty}\frac{1}{n+1}I(A^n\rightarrow{B}^n)\\
&=\lim_{n\longrightarrow\infty}\frac{1}{2}\frac{1}{n+1}\sum_{t=0}^n\log(1+\mathbb{E}\{A_t\}^2Var(Z_t)^{-1})\\
&=\lim_{n\longrightarrow\infty}\frac{1}{2}\frac{1}{n+1}\sum_{t=0}^n\log(1+\frac{P_t}{Q})\\
&=\lim_{n\longrightarrow\infty}\frac{1}{2}\frac{1}{n+1}\sum_{t=0}^n\log\frac{|\Lambda_t|}{|\Delta_t|}\tri\frac{1}{2}\log\frac{|\Lambda_\infty|}{|\Delta_\infty|}.
\end{align*}
We can design $\{({\cal A}_t,{\cal B}_t):~t\in\mathbb{N}\}$, by introducing the nonnegative components $\alpha_1,\ldots,\alpha_p$, $\sum_{i=1}^p{\alpha}_i=1$,~$i=2,\ldots,p$, and by considering the following transformations
\begin{equation}
{\cal A}_t=\Big{[}\sqrt{\frac{\alpha_1{P}_t}{\lambda_{t,1}}},\ldots,\sqrt{\frac{\alpha_p{P}_t}{\lambda_{t,p}}}\Big{]},~
{\cal B}_t=\Big{[}\sqrt{\alpha_1{P}_t\lambda_{t,1}},\ldots,\sqrt{\alpha_p{P}_t\lambda_{t,p}}\Big{]}^{tr},~t\in\mathbb{N}.\label{equation54}
\end{equation}
Define
\begin{align}
H_t&={\cal B}_t{\cal A}_t=\Big{[}\sqrt{\alpha_1{P}_t\lambda_{t,1}},\ldots,\sqrt{\alpha_p{P}_t\lambda_{t,p}}\Big{]}^{tr}\Big{[}\sqrt{\frac{\alpha_1{P}_t}{\lambda_{t,1}}},\ldots,\sqrt{\frac{\alpha_p{P}_t}{\lambda_{t,p}}}\Big{]}\nonumber\\
&= \left[ \begin{array}{c}
\sqrt{\alpha_1{P}_t\lambda_{t,1}} \\
\ldots\\
\sqrt{\alpha_p{P}_t\lambda_{t,p}} \end{array} \right]\Big{[}\sqrt{\frac{\alpha_1{P}_t}{\lambda_{t,1}}},\ldots,\sqrt{\frac{\alpha_p{P}_t}{\lambda_{t,p}}}\Big{]}\nonumber\\
&= P_t\left[ \begin{array}{ccc}
\alpha_1 & \ldots  \sqrt{\alpha_1\alpha_p\frac{\lambda_{t,1}}{\lambda_{t,p}}} \\
\vdots &   \vdots\\
\sqrt{\alpha_p\alpha_1 \frac{\lambda_{t,p}}{\lambda_{t,1}}} & \ldots  \alpha_p\end{array} \right]\in\mathbb{R}^{p\times{p}}.\nonumber
\end{align}
Therefore,
\begin{equation}
\tilde{\Gamma}_t={\cal B}_t{\cal A}_tE_tK_t+{\cal B}_tZ_t,~\Gamma_t=E_tK_t,~t\in\mathbb{N}.\label{equation56}
\end{equation}
By pre-multiplying $\tilde{\Gamma}_t$ by $E_t^{tr}$ we can construct
\begin{eqnarray*}
\tilde{K}_t&=&E_t^{tr}\tilde{\Gamma}_t\\
&=&E_t^{tr}{\cal B}_t{\cal A}_tE_tK_t+E_t^{tr}{\cal B}_tZ_t,~t\in\mathbb{N}.
\end{eqnarray*}
The reproduction of $Y_t$ is given by the sum of $\tilde{K}_t$ and $C\widehat{X}_{t|t-1}$ as follows.
\begin{eqnarray}
\tilde{Y}_t&=&\Psi_t(B^t,\tilde{Y}^{t-1})\nonumber\\
&=&\tilde{K}_t+C\widehat{X}_{t|t-1},~\hat{X}_t=\mathbb{E}\Big{\{}X_t|\sigma\{\tilde{Y}^{t-1}\}\Big{\}}\label{eq.13}\\
&=&E_t^{tr}{\cal B}_t{\cal A}_tE_tK_t+E_t^{tr}{\cal B}_tZ_t+C\widehat{X}_{t|t-1},~t\in\mathbb{N}.\label{eq.14}
\end{eqnarray}
Next, it will be shown that the desired distortion is achieved by the above realization while the filter of $\{Y_t:~t\in\mathbb{N}\}$ is based on $\{\tilde{Y}_t:~t\in\mathbb{N}\}$ given by (\ref{eq.14}).\\
First, we notice that
\begin{eqnarray}
\lim_{n\longrightarrow\infty}\mathbb{E}\Big{\{}(Y_t-\tilde{Y}_t)^{tr}(Y_t-\tilde{Y}_t)\Big{\}}=\lim_{n\longrightarrow\infty}Trace\Big{(}\mathbb{E}\Big{\{}(Y_t-\tilde{Y}_t)(Y_t-\tilde{Y}_t)^{tr}\Big{\}}\Big{)}\nonumber
\end{eqnarray}
Then we can compute 
\begin{align}
&\mathbb{E}\Big{\{}(Y_t-\tilde{Y}_t)^{tr}(Y_t-\tilde{Y}_t)\Big{\}}=Trace\mathbb{E}\Big{\{}(K_t-\tilde{K}_t)(K_t-\tilde{K}_t)^{tr}\Big{\}}\nonumber\\
&=Trace\mathbb{E}\Big{\{}(K_t-E_t^{tr}\tilde{\Gamma}_t)(K_t-E_t^{tr}\tilde{\Gamma}_t)^{tr}\Big{\}}\nonumber\\
&=Trace\mathbb{E}\Big{\{}(K_t-E_t^{tr}{\cal B}_t{\cal A}_tE_tK_t-E_t^{tr}{\cal B}_tZ_t)(K_t-E_t^{tr}{\cal B}_t{\cal A}_tE_tK_t-E_t^{tr}{\cal B}_tZ_t)^{tr}\Big{\}}\nonumber\\
&=Trace\mathbb{E}\Big{\{}\big{(}(I-E_t^{tr}{\cal B}_t{\cal A}_tE_t)K_t-E_t^{tr}{\cal B}_tZ_t\big{)}\big{(}(I-E_t^{tr}{\cal B}_t{\cal A}_tE_t)K_t-E_t^{tr}{\cal B}_tZ_t\big{)}^{tr}\Big{\}}\nonumber\\
&=Trace\Big{\{}(I-E_t^{tr}{\cal B}_t{\cal A}_tE_t)\Lambda_t(I-E_t^{tr}{\cal B}_t{\cal A}_tE_t)^{tr}+E_t^{tr}{\cal B}_tQ{\cal B}_t^{tr}E_t\Big{\}}\nonumber\\
&=Trace\Big{\{}(I-E_t^{tr}{\cal B}_t{\cal A}_tE_t)E_t^{tr}diag(\lambda_{t,1},\ldots,\lambda_{t,p})E_t(I-E_t^{tr}{\cal B}_t{\cal A}_tE_t)^{tr}+E_t^{tr}{\cal B}_tQ{\cal B}_t^{tr}E_t\Big{\}}\nonumber\\
&=Trace\Big{\{}E_t^{tr}\Big{(}(I-{\cal B}_t{\cal A}_t)diag(\lambda_{t,1},\ldots,\lambda_{t,p})(1-{\cal B}_t{\cal A}_t)^{tr}+({\cal B}_tQ{\cal B}_t^{tr})\Big{)}E_t\Big{\}}\nonumber\\
&=Trace\Big{\{}E_t^{tr}diag(\delta_{t,1},\ldots,\delta_{t,p})E_t\Big{\}}=\sum_{i=1}^p\delta_{t,i}=D.\label{equation57}
\end{align}
Hence, $\lim_{n\longrightarrow\infty}\frac{1}{n+1}\sum_{t=0}^n\mathbb{E}\Big{\{}(Y_t-\tilde{Y}_t)^{tr}(Y_t-\tilde{Y}_t)\Big{\}}=\lim_{n\longrightarrow\infty}\frac{1}{n+1}\sum_{t=0}^n\sum_{i=1}^p\delta_{t,i}=\sum_{i=1}^p\delta_{\infty,i}=D$.\\
Thus, by substituting the values of ${\cal B}_t, {\cal A}_t$ in terms of  $\{\delta_{t,i}\}_{i=1}^p, \{\lambda_{t,i}\}_{i=1}^p, P_t, Q$, and taking the limit in (\ref{equation57}) for $P_{\infty}\tri\lim_{t\longrightarrow\infty}P_t$, we get the general equation 
\begin{align}
\sum_{i=1}^p\Big[(1-\alpha_iP_\infty)\lambda_{\infty,i}(1-\alpha_i{P}_\infty)+a_iP_\infty{Q}\lambda_{\infty,i}\Big]=\sum_{i=1}^p\delta_{\infty,i}=D.\label{equation58}
\end{align} 
Therefore, to complete the realization we need to calculate $\{a_\infty\}_{i=1}^p$ in terms of the known eigenvalues $\{\lambda_{\infty,i},\delta_{\infty,i}\}_{i=1}^p$, the constants of the power level $P_\infty$, and channel's noise variance $Q$.  Note that due to the solution of the nonanticipative RDF for multidimensional partially observed Gaussian source, the encoding is performed only when $\frac{\{\lambda_{\infty,i}\}_{i=1}^p}{\{\delta_{\infty,i}\}_{i=1}^p}>1$. For more than one active modes of transmission $\{\lambda_{\infty,i}:~i=1,\ldots,k,~2\leq{k}\leq{p}\}$ then $\sum_{i=1}^k\xi_{\infty}=D\Longrightarrow\xi_{\infty}=\frac{D}{k}$.\\
We demonstrate this for the case where two active modes $\lambda_{\infty,1},~\lambda_{\infty,2}$, transmitted over the scalar channel, i.e., $\sum_{i=1}^2\alpha_i=\alpha_1+\alpha_2=1$, $\alpha_1\geq{0}$, $\alpha_2\geq{0}$.
\noi For $k=2$, (\ref{equation58}) is simplified as:\\
\begin{align}
\sum_{i=1}^2\Big[(1-\alpha_iP_\infty)\lambda_{\infty,i}(1-\alpha_iP_\infty)+a_iP_{\infty}Q\lambda_{\infty,i}\Big]=\sum_{i=1}^2\delta_{\infty,i}=D
\label{equation60}
\end{align}
with the following encoding requirement
\begin{align}
\frac{\lambda_{\infty,1}}{\xi_{\infty}}>1,~\frac{\lambda_{\infty,2}}{\xi_{\infty}}>1\Longrightarrow\frac{\lambda_{\infty,1}+\lambda_{\infty,2}}{D}>1.
\end{align} 
\noi After some calculations, (\ref{equation60}) is simplified in the following second order equation
\begin{align}
\alpha_2^2\Big[(\lambda_{\infty,1}+\lambda_{\infty,2})P_\infty^2\Big]+\alpha_2\Big[P_\infty\big((\lambda_{\infty,1}-\lambda_{\infty,2})(2-Q)-2\lambda_{\infty,1}P_t\big)\Big]\nonumber\\
+(\lambda_{\infty,1}+\lambda_{\infty,2})-D+\lambda_{\infty,1}P_\infty\Big[P_\infty+Q-1\Big]=0.\label{equation61}
\end{align} 
This quadratic equation can be solved numerically simultaneously with the equation of the filter, i.e., see (\ref{11}).\\
\noi{\it Decoder.} The decoder (mean square estimator) is $\tilde{Y}_t=\tilde{K}_t+C\widehat{X}_{t|t-1}$, where $\widehat{X}_{t|t-1}:~t\in\mathbb{N}$ is obtained from the modified Kalman filter as follows. 
Recall that
\begin{eqnarray}
\tilde{Y}_t&=&\tilde{K}_t+C\widehat{X}_{t|t-1}\nonumber\\
&=&E_\infty^{tr}H_\infty{E}_\infty(Y_t-C\widehat{X}_{t|t-1})+E_\infty^{tr}{\cal B}_\infty{Z}_t+C\widehat{X}_{t|t-1}\nonumber\\
&=&E_\infty^{tr}H_{\infty}E_{\infty}(CX_t+GV_t-C\widehat{X}_{t|t-1})+E_\infty^{tr}{\cal B}_{\infty}Z_t+C\widehat{X}_{t|t-1}\nonumber\\
&=&E_\infty^{tr}H_{\infty}E_{\infty}(CX_t-\widehat{X}_{t|t-1})+C\widehat{X}_{t|t-1}+(E_{\infty}^{tr}H_{\infty}E_{\infty}GV_t+E_{\infty}^{tr}{\cal B}_{\infty}Z_t)\nonumber
\end{eqnarray}
where $\{V_t:~t\in\mathbb{N}\}$ and $\{Z_t:~t\in\mathbb{N}\}$ are independent Gaussian vectors. Then $\widehat{X}_{t|t-1}=E\big{\{}X_t|\sigma\{\tilde{Y}^{t-1}\}\big{\}}$ is given by the modified Kalman filter
\begin{align}
\widehat{X}_{t+1|t-1}&=A\widehat{X}_{t|t-1}+A\Sigma_\infty(E_\infty^{tr}H_{\infty}E_{\infty}C)^{tr}M_{\infty}^{-1}\big(\tilde{Y}_t-C\widehat{X}_{t|t-1}\big),~\widehat{X}_0=\bar{x}_0\label{10}\\
\Sigma_{\infty}&=A\Sigma_\infty{A}^{tr}-A\Sigma_{\infty}(E_\infty^{tr}H_\infty{E}_{\infty}C)^{tr}M_{\infty}^{-1}(E_{\infty}^{tr}H_{\infty}E_{\infty}C)\Sigma_{\infty}A+BB_{\infty}^{tr}\label{11}
\end{align}
where 
\begin{align}
M_\infty&=E_\infty^{tr}H_\infty{E}_{\infty}C\Sigma_{\infty}(E_{\infty}^{tr}H_{\infty}E_{\infty}C)^{tr}+E_{\infty}^{tr}H_{\infty}E_{\infty}GG^{tr}(E_{\infty}^{tr}H_{\infty}E_{\infty})^{tr}\nonumber\\
&+E_{\infty}^{tr}{\cal B}_{\infty}Q{\cal B}_{\infty}^{tr}E_t^{tr}\nonumber
\end{align}
and $E_{\infty}$ is the unitary matrix that diagonalizes $\Lambda_{\infty}$ by
\begin{eqnarray}
E_{\infty}\Lambda_{\infty}E_{\infty}^{tr}=diag(\lambda_{\infty,1},\ldots,\lambda_{t,p}).\nonumber
\end{eqnarray}
\noi Finally, the matching of the source to the channel is obtained as follows. First, we remind that the power constraint satisfies $\mathbb{E}\{(A_t)^2\}={P}_t$,~$\lim_{t\rightarrow\infty}P_t=P_{\infty}\equiv{P}$.
\begin{align}
R^{na}(D)&=\lim_{n\rightarrow\infty}\inf_{\substack{P_{\tilde{Y}^n|Y^n}(d\tilde{y}^n|{y}^n)\\ \in{\cal Q}^c_{0,n}(D)}}\frac{1}{n+1}\mathbb{I}_{X^n\rightarrow{Y^n}}(P_{Y^n},\overrightarrow{P}_{\tilde{Y}^n|Y^n})\nonumber\\
&=\lim_{n\rightarrow\infty}\frac{1}{2}\frac{1}{n+1}\sum_{t=0}^n\sum_{i=1}^p\log\Big(\frac{\lambda_{t,i}}{\delta_{t,i}}\Big)\nonumber\\
&=\frac{1}{2}\sum_{i=1}^p\log\Big(\frac{\lambda_{\infty,i}}{\delta_{\infty,i}}\Big)\nonumber\\
&=\frac{1}{2}\log\frac{|\Lambda_{\infty}|}{|\Delta_{\infty}|}=\frac{1}{2}\log(1+\frac{P}{Q})=C(P).\nonumber
\end{align}

Thus, for a given $(D,P)$, $C(P)=R^{na}(D)$ is the minimum capacity under which there exists a realizable filter for the data reproduction of $\{Y_t:~t\in\mathbb{N}\}$ by $\{\tilde{Y}_t:~t\in\mathbb{N}\}$ ensuring an average distortion equal to $D$. The filter of $\{X_i:~i\in\mathbb{N}\}$ or $\{Y_i:~i\in\mathbb{N}\}$ is obtained for $\{\tilde{Y}_i:~i\in\mathbb{N}\}$ given by (\ref{equation56}) or the auxiliary data $B_i=A_i(Y_i,\tilde{Y}^{i-1})+Z_i$,~$i\in\mathbb{N}$. Finally, the filter is the steady state version of (\ref{10}), (\ref{11}).

\section{Conclusion}

In this paper, the solution of the nonanticipative RDF is obtained on abstract spaces using the topology of weak convergence of probability measures and a special case of directed information. A specific example that realizes the optimal causal filter is presented and the connection between nonanticipative RDF and source-channel matching via uncoded or symbol-by-symbol transmission is derived.

\bibliographystyle{model1-num-names}
\bibliography{photis_filtering_weak,photis_filtering_weakstar}

\end{document}